\documentclass[conference]{IEEEtran}

\usepackage{amsmath,amssymb,amsthm}
\usepackage{cite}
\usepackage{url}

\newtheorem{theorem}{Theorem}
\newtheorem{lemma}[theorem]{Lemma}
\newtheorem{corollary}[theorem]{Corollary}
\newtheorem{proposition}[theorem]{Proposition}
\theoremstyle{remark}
\newtheorem{remark}[theorem]{Remark}

\newcommand{\OPTfree}{\operatorname{OPT}_{\mathrm{free}}}
\newcommand{\OPTcard}{\operatorname{OPT}_{\mathrm{card}}}
\newcommand{\OPTshift}{\operatorname{OPT}_{D}^{-}}
\newcommand{\OPTstandard}{\operatorname{OPT}_{D}}

\title{Dasgupta’s Hierarchical Clustering Objective: Geometry and the Price of the Cardinality Constraint}

\author{
\IEEEauthorblockN{Peiyuan Sun}
\IEEEauthorblockA{Beihang University}
}

\begin{document}

\maketitle

\begin{abstract}
The cost of a hierarchical clustering can be represented by an
ultrametric whose lowest-common-ancestor labels are cluster
cardinalities.  We relate this known representation to the shortest-path
geometry of a similarity graph.  For a connected support graph $G$, let
$d_G$ be its unit-length shortest-path metric and let the edge weights
enter only the objective.  We prove that the shifted Dasgupta optimum is
exactly the minimum edge-weighted cost of a cardinality-realizable
ultrametric that dominates $d_G$.  Connectedification lemmas put this
problem and its freely labeled dominating-ultrametric relaxation on the
same class of connected binary hierarchies, labeled respectively by
cardinality and graph diameter.  As a sharp baseline, we determine the
exact worst-case price of cardinality realizability: on every
$n$-vertex instance the ratio of the two optima is at most
$(2n-1)/3$, with equality on the unweighted complete graph; the sharp
factor for the standard unshifted objective is $2(n+1)/3$.  Our
principal structural result bounds this gap by a hereditary weighted
  fragmentation profile defined through connected balanced cuts.  Uniform
  local control gives an $O(\log n)$ gap, polynomial decay gives a constant
  gap, and the logarithmic order is tight even for unweighted trees of
  maximum degree $3$.  On locally regular bounded-degree trees, the
  hierarchy can be constructed in $O(n\log n)$ time.  An energy
  decomposition and a geometric density bound provide
supporting instance-sensitive estimates.  Thus the cardinality label has
an unavoidable linear worst case but admits substantially smaller bounds
on natural sparse graph classes.
\end{abstract}

\begin{IEEEkeywords}
hierarchical clustering, Dasgupta cost, ultrametric, graph metric,
non-contractive embedding, weighted fragmentation, balanced cut,
average distortion
\end{IEEEkeywords}

\section{Introduction}
\label{sec:introduction}

Dasgupta's objective turns similarity-based hierarchical clustering into
a precise combinatorial optimization problem: an edge $xy$ pays its
similarity weight multiplied by the size of the smallest cluster that
contains both endpoints~\cite{Dasgupta2016}.  More generally, dendrograms
admit an ultrametric representation, a viewpoint used by Carlsson and
M\'emoli to formulate axiomatic, stability, and convergence results for
hierarchical clustering~\cite{CarlssonMemoli2010}.  In Dasgupta's model,
however, the label at the lowest common ancestor of $x$ and $y$ is fixed
to be the cardinality of the corresponding cluster.  Roy and Pokutta
established an exact correspondence between hierarchies and these
restricted ultrametrics~\cite{RoyPokutta2017}.  This paper asks a
different geometric question.
If the nonzero similarities define a connected support graph with
shortest-path metric $d_G$, how costly is the requirement that the
labels of a dominating ultrametric be cluster cardinalities rather than
freely chosen geometric scales?

Moseley and Wang studied an affine-dual reward formulation of
Dasgupta's objective and used it to analyze average linkage, divisive
local search, random hierarchical clustering, and several classical
heuristics~\cite{MoseleyWang2023}.  Their work asks how particular
algorithms compare with the optimal hierarchy.  Our question is
orthogonal: we compare the optimal cardinality-labeled hierarchy with a
freely labeled dominating-ultrametric relaxation.  The precise affine
relation between the two objective normalizations is recorded in
Section~\ref{sec:related}.

The geometry used here is deliberately the unweighted shortest-path
geometry of the support graph.  Similarity magnitudes determine the
measure with which stretch is averaged, but not the underlying metric.
Accordingly, our interpretation concerns the interaction between
support-graph geometry and cardinality labels; it is not a claim that the
support metric is the unique metric representation of weighted
similarities.

The distinction is meaningful even though the objective is evaluated
only on support edges.  A freely labeled ultrametric may choose its
cluster scale from the graph diameter.  A Dasgupta hierarchy must use
$|S|-1$ for a cluster $S$.  These labels agree on paths but can differ
by a factor of order $|S|$ when many vertices lie in a small-diameter
region.  The comparison therefore isolates a structural restriction of
the Dasgupta objective, rather than the computational difficulty of
approximating that objective.
The free benchmark is a controlled relaxation: it keeps $d_G$, the
edge-weight measure, and the non-contraction requirement fixed, and
removes only the requirement that an internal label equal $|S|-1$.
The complete-graph separation below is therefore a baseline exposing the
limitation of this relaxation in dense support graphs, rather than evidence
that the support metric captures dense similarity geometry.

Our contributions are summarized by three main results.
\begin{enumerate}
    \item \emph{Geometric equivalence and common normal forms.}
    Lemma~\ref{lem:connectedification} converts every dominating
    ultrametric, without increasing its value on any support edge, into
    a connected binary hierarchy labeled by graph diameters.
    Lemma~\ref{lem:dasgupta-connectedification} gives the cardinality
    analogue for every starting hierarchy.  Consequently,
    Corollary~\ref{cor:dasgupta-geometric-equivalence} identifies the
    shifted Dasgupta optimum exactly with the minimum over
    cardinality-realizable ultrametrics that dominate $d_G$.

    \item \emph{Exact worst-case price of cardinality realizability.}
    Theorem~\ref{thm:sharp-gap} proves that every connected weighted graph
    on $n$ vertices satisfies
    \begin{equation}
    1
    \le
    \frac{\OPTcard(G,w)}{\OPTfree(G,w)}
    \le
    \frac{2n-1}{3}.
    \label{eq:introduction-gap}
    \end{equation}
    The unweighted complete graph attains equality for every $n$, and
    Corollary~\ref{cor:unshifted-sharp-gap} gives the sharp factor
    $2(n+1)/3$ for the standard, unshifted objective.

    \item \emph{Hereditary structural bounds and a tight sparse
    example.}
    Theorem~\ref{thm:hereditary-fragmentation} bounds the gap by the sum
    of a weighted fragmentation profile over cardinality scales.  A
    uniform profile gives an $O(\log n)$ gap, while polynomial scale
    decay gives a constant gap.  Locally regular bounded-degree weighted
    trees satisfy the uniform condition and admit an explicit recursive
    construction; Theorem~\ref{thm:binary-tree-log-gap} shows that the
    resulting $\Theta(\log n)$ order is tight even for unweighted trees
    of maximum degree $3$.
\end{enumerate}

Two secondary estimates make the comparison instance-sensitive.  A
cut-energy identity expresses the fixed-hierarchy ratio as the
free-energy-weighted mean of
$\bigl(|S|-1\bigr)/\operatorname{diam}_{d_G}(S)$ and yields an elementary
tail-to-first-moment estimate, analogous in form to scaling-to-average
arguments~\cite{ABN2011,ABN2015}.
A purely geometric corollary controls the gap by the maximum density of
a connected set, equivalently up to a factor of $2$ by linear ball
growth.  These results are consequences and proof tools.  The sharp
general gap is the baseline separation, while the hereditary
fragmentation theorem and its tight bounded-degree-tree example form the
principal structural contribution.

This paper does not provide a new approximation guarantee for average
linkage, single linkage, complete linkage, bisecting $k$-means, or any
other practical clustering heuristic.  Its approximation factors compare
two geometric optimization benchmarks, not an algorithm with the
Dasgupta optimum.

The paper is organized as follows.  Section~\ref{sec:related} positions
the comparison relative to prior ultrametric and hierarchical-clustering
formulations.  Section~\ref{sec:preliminaries} fixes the two optimization
problems.  Sections~\ref{sec:free-connectedification}--
\ref{sec:cardinality-equivalence} establish the common connected normal
form.  Section~\ref{sec:sharp-gap} proves the sharp gap, and
Section~\ref{sec:fragmentation} proves the hereditary bound and its tight
tree example.  Sections~\ref{sec:energy} and
\ref{sec:geometric-consequences} give supporting instance-sensitive
consequences.

\section{Related Work}
\label{sec:related}

Dasgupta introduced the objective, its behavior on disconnected graphs
and cliques, and a recursive sparsest-cut approximation
\cite{Dasgupta2016}.  Roy and Pokutta characterized the
\emph{non-trivial ultrametrics} induced by cardinality labels and used
that characterization in an $O(\log n)$ LP-based approximation
\cite{RoyPokutta2017}.  Thus the hierarchy--cardinality-ultrametric
correspondence is prior work.  Charikar and Chatziafratis and,
independently, Cohen-Addad et al. obtained
$O(\sqrt{\log n})$-type approximation guarantees through recursive
sparsest cut and convex relaxations; they also developed hardness and
axiomatic results for the objective
\cite{CharikarChatziafratis2017,CohenAddadEtAl2019}.

The metric-geometric viewpoint on hierarchical clustering predates these
objective functions.  Carlsson and M\'emoli represent dendrograms as
ultrametric spaces and use the Gromov--Hausdorff distance to establish
axiomatic characterization, stability, and convergence results for
single linkage~\cite{CarlssonMemoli2010}.  Their theory concerns the
behavior of a clustering method under perturbation and sampling.  It
does not compare objective values, impose Dasgupta's cardinality labels,
or optimize a similarity-weighted distortion.  Our use of ultrametrics
therefore shares their geometric representation but addresses a
different question about the feasible label set of an objective.

Moseley and Wang introduced the reward formulation in a preliminary
NeurIPS 2017 paper and subsequently gave an expanded analysis
\cite{MoseleyWang2023}.  Write $W:=\sum_{e\in E}w_e$ and let
$\operatorname{cost}_{D}^{-}$ denote the shifted cost formally defined
in Section~\ref{sec:preliminaries}.  In our notation their objective is
\begin{equation}
\operatorname{rew}_G(T)
:=
\sum_{xy\in E}w_{xy}\bigl(n-|C_T(x,y)|\bigr).
\label{eq:moseley-wang-reward}
\end{equation}
Consequently,
\begin{equation}
\operatorname{rew}_G(T)+\operatorname{cost}_{D}^{-}(T)
=(n-1)W.
\label{eq:reward-shift-relation}
\end{equation}
Thus reward maximization, shifted-cost minimization, and the standard
Dasgupta problem have the same optimal hierarchies.  Multiplicative
approximation ratios are not preserved by this affine transformation:
Moseley and Wang prove a $1/3$ reward guarantee for average linkage while
also constructing polynomial cost lower bounds for classical linkage
methods.  They additionally show that their random recursive hierarchy
has exact expected reward $(n-2)W/3$, which by
\eqref{eq:reward-shift-relation} is equivalent to expected shifted cost
$(2n-1)W/3$.  This is the reward-side antecedent of the numerical
identity used in our sharp-gap upper bound.  Our contribution is not a
new analysis of that random algorithm; it uses the identity to obtain an
exact separation from the free dominating-ultrametric benchmark and then
develops graph-structural conditions under which that separation is much
smaller.

Charikar, Chatziafratis, and Niazadeh subsequently proved that the
$1/3$ guarantee is worst-case tight for average linkage on the
similarity-reward objective.  They established the analogous tight
$2/3$ bound for the dissimilarity objective and gave algorithms that
strictly beat the two average-linkage baselines
\cite{CharikarChatziafratisNiazadeh2019}.  These results separate the
quality of a particular linkage heuristic from the approximability of a
fixed hierarchical objective.  Our comparison is on a different axis:
both sides are optimized, while the admissible ultrametric labels change
from free scales to cluster cardinalities.

For unweighted graphs, H\o gemo, Paul, and Telle proved NP-completeness
and developed a normalization procedure~\cite{HogemoEtAl2020}.
H\o gemo et al. subsequently proved the existence of an optimal
hierarchy with connected clusters and related the objective on trees to
edge- and vertex-partition parameters; they also obtained a
$2$-approximation for the standard Dasgupta objective on unweighted
trees~\cite{HogemoEtAl2021}.  That approximation compares an algorithmic
hierarchy directly with the optimum inside the usual cardinality-labeled
hierarchy class.  By contrast, our tree construction is a witness for
the gap between the cardinality-realizable and free
dominating-ultrametric benchmarks: it applies to locally regular weighted
bounded-degree trees and bounds shifted cost relative to $\OPTfree$.
Its $O(\log n)$ factor is therefore not claimed as an improvement over
the $2$-approximation.  Our cardinality connectedification is consistent
with their connected-optimum result, but is stated edgewise for every
starting hierarchy and for arbitrary positive similarity weights.  The
additional normal form for a freely labeled dominating ultrametric has
no counterpart in those graph-parameter results.

Arutyunova and R\"oglin study a different ``price of hierarchy'': they
compare one nested sequence with separately optimal flat clusterings at
every level for radius and diameter objectives~\cite{ArutyunovaRoglin2025}.
Their benchmark, objectives, and feasibility constraint differ from the
single Dasgupta objective considered here.  Our price is instead the loss
caused by cardinality realizability relative to a freely labeled
dominating-ultrametric relaxation on the same support graph.

Tree and ultrametric embeddings traditionally allow internal labels to
be chosen freely.  Bartal initiated probabilistic tree approximation
\cite{Bartal1996}, and Fakcharoenphol, Rao, and Talwar obtained the tight
$O(\log n)$ expected distortion guarantee for dominating tree metrics
\cite{FRT2004}.  Abraham, Bartal, and Neiman developed scaling and
average-distortion guarantees, including scaling distortion
$O(\sqrt{1/\varepsilon})$ into ultrametrics and, as a consequence,
constant average distortion~\cite{ABN2011,ABN2015}.  Cohen-Addad et al.
study the complementary computational problem of
efficiently approximating the best worst-pair-distortion ultrametric for
high-dimensional Euclidean data~\cite{CohenAddadKarthikLagarde2020}.
These results motivate our free relaxation, but their distortion
criteria concern metric pairs and permit freely chosen scales; they do
not require a label to equal the cardinality below a tree node.  Our
weights are instead similarities supported on graph edges, and the
unweighted support metric supplies the non-contraction constraint.
Metric-fitting objectives, such as that of Ailon and
Charikar~\cite{AilonCharikar2011}, are also different: they approximate
supplied dissimilarities by a tree metric, whereas our numerator is the
similarity-weighted Dasgupta objective.

Table~\ref{tab:comparison} summarizes the boundary.  We do not propose a
new approximation algorithm for the Dasgupta optimum.  We instead
compare its cardinality-realizable geometry with the freely labeled
dominating-ultrametric benchmark and identify exact worst-case and
structured-instance gaps.

\begin{table*}[t]
\caption{Relation to the closest lines of work.}
\label{tab:comparison}
\centering
\footnotesize
\renewcommand{\arraystretch}{1.08}
\begin{tabular}{@{}p{0.15\textwidth}p{0.25\textwidth}p{0.27\textwidth}p{0.25\textwidth}@{}}
\hline
Work & Feasible object & Main focus & Difference from the present comparison \\
\hline
Dasgupta; Roy--Pokutta
& Binary hierarchies and cardinality-induced ultrametrics
& Objective properties, ultrametric characterization, and approximation
& Does not compare cardinality labels with freely chosen dominating labels \\
Carlsson--M\'emoli
& Dendrograms represented as ultrametric spaces
& Axiomatic characterization, stability, and convergence of clustering methods
& No Dasgupta objective, similarity weights, or cardinality-realizability gap \\
Moseley--Wang; Charikar et al.
& Hierarchies evaluated by affine-dual reward and Dasgupta cost
& Average-linkage guarantees and tight examples; algorithms beating average linkage
& No free dominating-ultrametric benchmark or cardinality-realizability gap \\
H\o gemo et al.
& Hierarchies of unweighted graphs
& Normalization, connected optimal clusters, graph parameters, and a tree
$2$-approximation
& No free dominating-ultrametric benchmark or weighted edgewise gap \\
Bartal; FRT; ABN; Cohen-Addad et al.
& Freely labeled dominating tree metrics or ultrametrics
& Worst-case, probabilistic, scaling, average, and efficient-computation guarantees
& No cardinality-realizability constraint and a different pair measure \\
This paper
& Cardinality-realizable versus free dominating ultrametrics on one support graph
& Exact worst-case gap and hereditary fragmentation bounds
& Similarity weights enter the objective; support edges define $d_G$ \\
\hline
\end{tabular}
\end{table*}

\section{Preliminaries}
\label{sec:preliminaries}

Let $G=(V,E,w)$ be a finite connected graph with $|V|=n\ge 2$, where
$w:E\to\mathbb{R}_{>0}$ is a similarity weight.  The metric
$d_G$ is the unit-length shortest-path metric of the support graph
$(V,E)$; the values $w_e$ occur only in the objective and are not
interpreted as edge lengths.

Write
\begin{equation}
W:=\sum_{e\in E}w_e.
\label{eq:total-weight}
\end{equation}

An \emph{ultrametric} on $V$ is a metric
$u:V\times V\to\mathbb{R}_{\ge 0}$ satisfying
\begin{equation}
u(x,y)\le \max\{u(x,z),u(z,y)\}
\label{eq:ultrametric-inequality}
\end{equation}
for all $x,y,z\in V$.  It \emph{dominates} $d_G$ if
\begin{equation}
u(x,y)\ge d_G(x,y)
\qquad \text{for all }x,y\in V.
\label{eq:domination}
\end{equation}
The free dominating-ultrametric value is
\begin{equation}
\OPTfree(G,w)
:=
\min_{\substack{u\text{ ultrametric on }V\\u\ge d_G}}
\sum_{xy\in E}w_{xy}u(x,y).
\label{eq:free-optimum}
\end{equation}
The existence of a minimizer will also follow from the connected normal
form proved below.

A \emph{hierarchical tree} $T$ is a rooted tree whose leaves are the
vertices of $V$.  Every node is identified with the set of leaves in
its subtree.  For $x,y\in V$, let $C_T(x,y)$ denote the cluster
associated with their lowest common ancestor.  The hierarchy is called
\emph{connected} if every cluster $S$ induces a connected subgraph
$G[S]$.  Let $\mathcal T_{\mathrm{bin}}(V)$ be the set of all rooted
binary hierarchies with leaf set $V$.

For a binary hierarchy $T$, the standard and shifted Dasgupta objectives
are, respectively,
\begin{align}
\operatorname{cost}_{D}(T)
&:=
\sum_{xy\in E}w_{xy}|C_T(x,y)|,
\label{eq:standard-dasgupta-definition}\\
\operatorname{cost}_{D}^{-}(T)
&:=
\sum_{xy\in E}w_{xy}\bigl(|C_T(x,y)|-1\bigr).
\label{eq:shifted-dasgupta-definition}
\end{align}
Their unrestricted optimal values are denoted by $\OPTstandard(G,w)$ and
$\OPTshift(G,w)$.  Since the additive difference is independent of the
tree,
\begin{equation}
\OPTstandard(G,w)=\OPTshift(G,w)+W.
\label{eq:opt-standard-shifted}
\end{equation}

For a connected hierarchy $T$, define its diameter labeling by
\begin{equation}
u_T^{\mathrm{diam}}(x,y)
:=
\begin{cases}
0, & x=y,\\
\operatorname{diam}_{d_G}\bigl(C_T(x,y)\bigr), & x\ne y.
\end{cases}
\label{eq:diameter-labeling}
\end{equation}
Because graph diameter is monotone under set inclusion,
$u_T^{\mathrm{diam}}$ is an ultrametric.  Moreover,
\begin{equation}
d_G(x,y)
\le
\operatorname{diam}_{d_G}\bigl(C_T(x,y)\bigr),
\label{eq:diameter-dominates}
\end{equation}
so it dominates $d_G$.

\section{Connectedification of the Free Ultrametric}
\label{sec:free-connectedification}

\begin{lemma}[Connectedification of a dominating ultrametric]
\label{lem:connectedification}
Let $u$ be an ultrametric on $V$ satisfying $u\ge d_G$.
Then there exists a rooted binary hierarchical tree $T$ on $V$
such that:
\begin{enumerate}
    \item every cluster of $T$ induces a connected subgraph of $G$;
    \item $u_T^{\mathrm{diam}}$ is an ultrametric dominating $d_G$;
    \item for every edge $xy\in E$,
    \begin{equation}
    u_T^{\mathrm{diam}}(x,y)\le u(x,y).
    \label{eq:edgewise-comparison}
    \end{equation}
\end{enumerate}
Consequently,
\begin{equation}
\sum_{xy\in E}w_{xy}u_T^{\mathrm{diam}}(x,y)
\le
\sum_{xy\in E}w_{xy}u(x,y).
\label{eq:cost-comparison}
\end{equation}
\end{lemma}

\begin{IEEEproof}
Let
\begin{equation}
\Lambda
:=
\{0\}\cup\{u(x,y):x,y\in V,\ x\ne y\}
\end{equation}
be the finite set of values attained by $u$.  For each
$t\in\Lambda$, define
\begin{equation}
x\sim_t y
\quad\Longleftrightarrow\quad
u(x,y)\le t.
\label{eq:threshold-equivalence}
\end{equation}
The ultrametric inequality implies that $\sim_t$ is an equivalence
relation.  Let $\mathcal U_t$ denote its partition into equivalence
classes.

For every $U\in\mathcal U_t$, take the connected components of the
induced graph $G[U]$.  Let $\mathcal P_t$ be the partition of $V$
formed by all these components.  The family
$\{\mathcal P_t\}_{t\in\Lambda}$ is nested.  Indeed, suppose that
$s\le t$ and $A\in\mathcal P_s$.  The set $A$ is a connected
component of $G[U]$ for some $U\in\mathcal U_s$.  Since
$\mathcal U_s$ refines $\mathcal U_t$, there exists
$U'\in\mathcal U_t$ with $U\subseteq U'$.  The set $A$ is
connected in $G[U']$, and hence is contained in a unique connected
component of $G[U']$, which is a block of $\mathcal P_t$.

The partition $\mathcal P_0$ consists of singletons.  If
\begin{equation}
t_{\max}:=\max_{x,y\in V}u(x,y),
\end{equation}
then $\mathcal U_{t_{\max}}=\{V\}$.  Since $G$ is connected,
$\mathcal P_{t_{\max}}=\{V\}$.  Consequently, the distinct blocks
appearing in the nested family $\{\mathcal P_t\}_{t\in\Lambda}$,
ordered by inclusion, define a rooted hierarchical tree $T_0$.  By
construction, every cluster of $T_0$ is connected in $G$.

Label every cluster $S$ of $T_0$ by
\begin{equation}
\Delta(S):=\operatorname{diam}_{d_G}(S).
\end{equation}
If $S\subseteq S'$, then $\Delta(S)\le\Delta(S')$.  Hence the
labels are nondecreasing toward the root and the lowest-common-ancestor
rule defines an ultrametric $u_{T_0}^{\mathrm{diam}}$.  For every
pair $x,y\in V$, the cluster $C_{T_0}(x,y)$ contains both vertices,
so
\begin{equation}
d_G(x,y)
\le
\operatorname{diam}_{d_G}\bigl(C_{T_0}(x,y)\bigr)
=
u_{T_0}^{\mathrm{diam}}(x,y).
\end{equation}
Thus $u_{T_0}^{\mathrm{diam}}$ dominates $d_G$ on all pairs.

We next establish the comparison with $u$ on graph edges.  Fix an
edge $xy\in E$ and put $t=u(x,y)$.  The vertices $x$ and $y$
belong to the same class $U\in\mathcal U_t$.  Since $xy$ is itself
an edge of $G[U]$, they also belong to the same connected component
$K\in\mathcal P_t$.  Let
\begin{equation}
S:=C_{T_0}(x,y).
\end{equation}
The construction of $T_0$ gives $S\subseteq K\subseteq U$.  Hence,
for every $a,b\in S$,
\begin{equation}
d_G(a,b)\le u(a,b)\le t.
\end{equation}
Taking the maximum over $a,b\in S$ yields
\begin{equation}
u_{T_0}^{\mathrm{diam}}(x,y)
=\operatorname{diam}_{d_G}(S)
\le t
=u(x,y).
\label{eq:nonbinary-edge-comparison}
\end{equation}

It remains to make the hierarchy binary.  Let $S$ be a node of
$T_0$ with children $S_1,\ldots,S_k$, where $k>2$.  Contracting
each connected set $S_i$ in the connected graph $G[S]$ produces a
connected quotient graph.  Choose an ordering of the children such that
each partial union
\begin{equation}
R_j:=S_1\cup\cdots\cup S_j,
\qquad 2\le j\le k,
\end{equation}
is connected.  Such an ordering is obtained, for example, by rooting a
spanning tree of the quotient graph and listing each vertex after its
parent.  Replace the $k$-ary node by the binary chain
\begin{equation}
R_2\subset R_3\subset\cdots\subset R_k=S.
\end{equation}
Every new cluster is connected and contained in $S$; therefore
\begin{equation}
\operatorname{diam}_{d_G}(R_j)
\le
\operatorname{diam}_{d_G}(S).
\end{equation}
Thus the refinement cannot increase the diameter label of the lowest
common ancestor of any graph edge.  Applying this operation to all
nonbinary nodes produces a binary connected hierarchy $T$ satisfying
\eqref{eq:edgewise-comparison}.

Finally, multiply \eqref{eq:edgewise-comparison} by $w_{xy}\ge0$ and
sum over $xy\in E$ to obtain \eqref{eq:cost-comparison}.
\end{IEEEproof}

\section{A Normal Form for the Free Optimum}
\label{sec:free-normal-form}

Let $\mathcal T_{\mathrm{conn}}^{\mathrm{bin}}(G)$ denote the set of
rooted binary hierarchies on $V$ whose clusters are connected in
$G$.

\begin{corollary}[Connected diameter representation]
\label{cor:connected-diameter-representation}
The free dominating-ultrametric optimum satisfies
\begin{align}
\OPTfree(G,w)
&=
\min_{T\in\mathcal T_{\mathrm{conn}}^{\mathrm{bin}}(G)}
\nonumber\\
&\quad
\sum_{xy\in E}w_{xy}
\operatorname{diam}_{d_G}\bigl(C_T(x,y)\bigr).
\label{eq:connected-free-optimum}
\end{align}
In particular, the free optimum is attained by a connected,
diameter-normalized hierarchy.
\end{corollary}

\begin{IEEEproof}
For every $T\in\mathcal T_{\mathrm{conn}}^{\mathrm{bin}}(G)$, the
map $u_T^{\mathrm{diam}}$ is a feasible ultrametric in
\eqref{eq:free-optimum}.  Therefore the left-hand side of
\eqref{eq:connected-free-optimum} is at most the right-hand side.

Conversely, Lemma~\ref{lem:connectedification} transforms every
feasible ultrametric $u\ge d_G$ into a tree
$T\in\mathcal T_{\mathrm{conn}}^{\mathrm{bin}}(G)$ whose objective
value is no larger.  Hence the right-hand side is at most the left-hand
side.  The family of binary hierarchies on the finite set $V$ is
finite, so the right-hand minimum is attained.
\end{IEEEproof}

\section{Cardinality Connectedification and Equivalence}
\label{sec:cardinality-equivalence}

For a hierarchy $T$, define its shifted cardinality labeling by
\begin{equation}
h_T(x,y)
:=
\begin{cases}
0, & x=y,\\
|C_T(x,y)|-1, & x\ne y.
\end{cases}
\label{eq:cardinality-labeling}
\end{equation}
The labels are nondecreasing toward the root, so $h_T$ is an
ultrametric.  A cardinality-realizable ultrametric is any map $h_H$
arising from a hierarchy $H\in\mathcal T_{\mathrm{bin}}(V)$ through
\eqref{eq:cardinality-labeling}.  We define its natural dominating
optimum by
\begin{equation}
\OPTcard(G,w)
:=
\min_{\substack{H\in\mathcal T_{\mathrm{bin}}(V)\\h_H\ge d_G}}
\sum_{xy\in E}w_{xy}h_H(x,y).
\label{eq:cardinality-optimum}
\end{equation}
Here $h_H\ge d_G$ is required on all pairs.  The feasible family is
nonempty: every connected hierarchy has $h_H\ge d_G$, as proved below.
The next lemma shows that connectedness is a normal form rather than an
extra restriction on \eqref{eq:cardinality-optimum}.

\begin{lemma}[Connectedification of a Dasgupta hierarchy]
\label{lem:dasgupta-connectedification}
Let $H$ be an arbitrary rooted binary hierarchy on $V$.  There exists a
rooted binary hierarchy $T\in\mathcal T_{\mathrm{conn}}^{\mathrm{bin}}(G)$
such that, for every edge $xy\in E$,
\begin{equation}
|C_T(x,y)|\le |C_H(x,y)|.
\label{eq:cardinality-edgewise}
\end{equation}
Consequently,
\begin{equation}
\operatorname{cost}_{D}^{-}(T)
\le
\operatorname{cost}_{D}^{-}(H)
\quad\text{and}\quad
\operatorname{cost}_{D}(T)
\le
\operatorname{cost}_{D}(H).
\label{eq:dasgupta-connected-cost}
\end{equation}
\end{lemma}

\begin{IEEEproof}
Let $\mathcal L(H)$ be the laminar family of leaf sets represented by
the nodes of $H$.  Form a new collection
\begin{equation}
\mathcal C
:=
\left\{
K:
K\text{ is a connected component of }G[A]
\text{ for some }A\in\mathcal L(H)
\right\}.
\label{eq:component-laminar-family}
\end{equation}
The collection $\mathcal C$ is laminar.  Indeed, clusters in
$\mathcal L(H)$ are either disjoint or nested.  If $A\subseteq B$, then
every connected component of $G[A]$ is contained in a unique connected
component of $G[B]$.  Since $\mathcal C$ contains all singletons and,
because $G$ is connected, also contains $V$, its distinct members form
a rooted hierarchy $T_0$ under inclusion.  Every cluster of $T_0$ is
connected.

Fix an edge $xy\in E$ and let $A=C_H(x,y)$.  The edge $xy$ lies in
$G[A]$, so $x$ and $y$ belong to the same connected component $K$ of
$G[A]$.  The smallest member of $\mathcal C$ containing both endpoints
is contained in $K$, and hence
\begin{equation}
|C_{T_0}(x,y)|
\le |K|
\le |A|
=|C_H(x,y)|.
\label{eq:nonbinary-cardinality-comparison}
\end{equation}

The hierarchy $T_0$ need not be binary.  Suppose that a cluster $S$ has
maximal proper children $S_1,\ldots,S_k$.  These children partition $S$.
After contracting each $S_i$, the graph $G[S]$ induces a connected
quotient graph.  A parent-before-child traversal of a spanning tree of
this quotient orders the children so that every prefix union
\begin{equation}
R_j:=S_1\cup\cdots\cup S_j
\qquad (2\le j\le k)
\label{eq:cardinality-prefix-unions}
\end{equation}
is connected.  Replacing the $k$-ary node by the binary chain
$R_2\subset\cdots\subset R_k=S$ creates only connected clusters.
For endpoints originally separated among different children of $S$,
the new lowest common ancestor is a subset of $S$ and therefore has no
larger cardinality.  Applying this refinement at every nonbinary node
produces the required hierarchy $T$.  Inequality
\eqref{eq:dasgupta-connected-cost} follows from
\eqref{eq:cardinality-edgewise} and the nonnegativity of the weights.
\end{IEEEproof}

\begin{corollary}[Geometric form of the Dasgupta optimum]
\label{cor:dasgupta-geometric-equivalence}
The shifted Dasgupta optimum equals the natural cardinality-realizable
dominating-ultrametric optimum and admits the connected normal form
\begin{align}
\OPTshift(G,w)
&=\OPTcard(G,w)
\nonumber\\
&=
\min_{T\in\mathcal T_{\mathrm{conn}}^{\mathrm{bin}}(G)}
\sum_{xy\in E}w_{xy}h_T(x,y).
\label{eq:dasgupta-cardinality-equivalence}
\end{align}
Moreover, for every $T\in\mathcal T_{\mathrm{conn}}^{\mathrm{bin}}(G)$,
\begin{equation}
d_G(x,y)
\le
\operatorname{diam}_{d_G}\bigl(C_T(x,y)\bigr)
\le
|C_T(x,y)|-1
=h_T(x,y)
\label{eq:diameter-cardinality-chain}
\end{equation}
for all $x,y\in V$.  Thus $h_T$ is non-contractive on all pairs, not
only on the support edges.
\end{corollary}

\begin{IEEEproof}
Let $A$ be the minimum shifted cost over all binary hierarchies, let $B$
be the minimum in \eqref{eq:cardinality-optimum}, and let $C$ be the
minimum shifted cost over connected binary hierarchies.  The feasible
family defining $B$ is a subfamily of all binary hierarchies, so
$A\le B$.  If $G[S]$ is connected, then any two vertices of $S$ can be
joined inside $S$ by a simple path with at most $|S|-1$ edges.  Taking
$S=C_T(x,y)$ gives \eqref{eq:diameter-cardinality-chain}; hence every
connected binary hierarchy is feasible for $B$, and $B\le C$.
Lemma~\ref{lem:dasgupta-connectedification} transforms every binary
hierarchy into a connected binary hierarchy of no larger shifted cost,
so $C\le A$.  Thus $A=B=C$, proving
\eqref{eq:dasgupta-cardinality-equivalence} and the domination claim.
\end{IEEEproof}

\begin{remark}[Weighted average stretch]
\label{rem:average-stretch}
Because $d_G(x,y)=1$ for every $xy\in E$, the quantities
\begin{equation}
\frac{1}{W}\sum_{xy\in E}w_{xy}u(x,y)
\quad\text{and}\quad
\frac{1}{W}\operatorname{cost}_{D}^{-}(T)
\label{eq:weighted-average-stretch}
\end{equation}
are weighted average stretches on the support edges.  Corollary
\ref{cor:dasgupta-geometric-equivalence} therefore identifies the
shifted Dasgupta problem with minimum weighted average stretch into the
restricted class of cardinality-realizable dominating ultrametrics.
\end{remark}

\section{The Exact Cardinality-Realizability Gap}
\label{sec:sharp-gap}

We first record a tree identity that is responsible for the sharp
constant.  It is the shifted form of the familiar invariance of
Dasgupta's cost on the unweighted clique~\cite{Dasgupta2016}; the proof
is included because the exact expectation is used below.

\begin{lemma}[Pair-cardinality identity]
\label{lem:pair-cardinality-identity}
Let $\tau$ be any rooted binary tree with $n$ leaves.  Then
\begin{equation}
Q(\tau)
:=
\sum_{\{i,j\}}
\bigl(|C_\tau(i,j)|-1\bigr)
=
\frac{n(n-1)(2n-1)}{6},
\label{eq:pair-cardinality-identity}
\end{equation}
where the sum is over all unordered pairs of distinct leaves.  In
particular, $Q(\tau)$ depends only on $n$ and not on the shape of
$\tau$.
\end{lemma}

\begin{IEEEproof}
The claim is immediate for $n=1$.  Suppose that the two subtrees below
the root contain $a$ and $b$ leaves, where $a+b=n$.  Each of the $ab$
crossing pairs has the root cluster as its lowest common ancestor and
therefore contributes $n-1$.  By induction,
\begin{align}
Q(\tau)
&=
\frac{a(a-1)(2a-1)}{6}
+
\frac{b(b-1)(2b-1)}{6}
\nonumber\\
&\quad
+ab(n-1)
\nonumber\\
&=
\frac{n(n-1)(2n-1)}{6}.
\label{eq:pair-cardinality-induction}
\end{align}
The last equality is a direct expansion using $n=a+b$.
\end{IEEEproof}

\begin{theorem}[Sharp worst-case gap]
\label{thm:sharp-gap}
For every connected weighted graph $G=(V,E,w)$ with $|V|=n\ge2$,
\begin{equation}
1
\le
\frac{\OPTcard(G,w)}{\OPTfree(G,w)}
\le
\frac{2n-1}{3}.
\label{eq:sharp-gap}
\end{equation}
The upper factor is best possible for every $n$: equality holds for the
unweighted complete graph $K_n$.
\end{theorem}

\begin{IEEEproof}
Every cardinality-realizable dominating ultrametric is feasible for the
free problem.  Hence $\OPTfree(G,w)\le\OPTcard(G,w)$, proving the first
inequality.

For the upper bound, fix any rooted binary tree shape $\tau$ with $n$
leaves and choose a uniformly random bijection $\pi:V\to L(\tau)$.
For a fixed edge $xy\in E$, the unordered pair
$\{\pi(x),\pi(y)\}$ is uniformly distributed over all unordered pairs of
leaves.  Lemma~\ref{lem:pair-cardinality-identity} gives
\begin{equation}
\mathbb E_\pi
\left[
|C_{\tau,\pi}(x,y)|-1
\right]
=
\frac{Q(\tau)}{\binom{n}{2}}
=
\frac{2n-1}{3}.
\label{eq:expected-cardinality-label}
\end{equation}
By linearity of expectation,
\begin{equation}
\mathbb E_\pi
\left[
\operatorname{cost}_{D}^{-}(\tau,\pi)
\right]
=
\frac{2n-1}{3}W.
\label{eq:expected-dasgupta-cost}
\end{equation}
By \eqref{eq:reward-shift-relation}, this expectation is also the
cost-side form of Moseley and Wang's exact random-hierarchy reward
calculation~\cite{MoseleyWang2023}.  The fixed-shape random-labeling
argument above supplies the particular formulation needed here.
There is therefore a deterministic labeling with cost at most the
right-hand side of \eqref{eq:expected-dasgupta-cost}.  Applying
Lemma~\ref{lem:dasgupta-connectedification} to that hierarchy yields a
connected binary hierarchy with no larger cost.  Consequently,
\begin{equation}
\OPTcard(G,w)
\le
\frac{2n-1}{3}W.
\label{eq:cardinality-versus-total-weight}
\end{equation}

For every feasible free ultrametric $u$ and every edge $xy\in E$,
domination gives $u(x,y)\ge d_G(x,y)=1$.  Hence
\begin{equation}
\sum_{xy\in E}w_{xy}u(x,y)\ge W,
\qquad
\OPTfree(G,w)\ge W.
\label{eq:free-lower-total-weight}
\end{equation}
Combining \eqref{eq:cardinality-versus-total-weight} and
\eqref{eq:free-lower-total-weight} proves the upper inequality in
\eqref{eq:sharp-gap}.

It remains to prove sharpness.  On the unweighted complete graph, the
constant ultrametric $u(x,y)=1$ for $x\ne y$ dominates $d_{K_n}$.
Therefore
\begin{equation}
\OPTfree(K_n)=\binom{n}{2}.
\label{eq:complete-free}
\end{equation}
Every cluster in $K_n$ is connected, and
Lemma~\ref{lem:pair-cardinality-identity} shows that every binary
hierarchy has shifted cost
\begin{equation}
\frac{n(n-1)(2n-1)}{6}
=
\frac{2n-1}{3}\binom{n}{2}.
\label{eq:complete-cardinality}
\end{equation}
Thus equality holds in \eqref{eq:sharp-gap}.
\end{IEEEproof}

\begin{corollary}[Sharp factor for the standard objective]
\label{cor:unshifted-sharp-gap}
For the standard Dasgupta objective,
\begin{equation}
\frac{\OPTstandard(G,w)}{\OPTfree(G,w)}
\le
\frac{2(n+1)}{3}.
\label{eq:unshifted-sharp-gap}
\end{equation}
The factor is best possible and is attained by the unweighted complete
graph.
\end{corollary}

\begin{IEEEproof}
Equations~\eqref{eq:opt-standard-shifted} and
\eqref{eq:cardinality-versus-total-weight} imply
\begin{equation}
\OPTstandard(G,w)
\le
\left(\frac{2n-1}{3}+1\right)W
=
\frac{2(n+1)}{3}W.
\end{equation}
Now apply \eqref{eq:free-lower-total-weight}.  Equality for $K_n$
follows from \eqref{eq:complete-free},
\eqref{eq:complete-cardinality}, and
$\OPTstandard(K_n)=\OPTshift(K_n)+\binom{n}{2}$.
\end{IEEEproof}

\section{Hereditary Weighted Fragmentation}
\label{sec:fragmentation}

The sharp gap theorem is unconditional, but its linear factor is driven
by connected regions that cannot be divided by a light, balanced cut.
We formalize the opposite regime using a profile defined directly on the
input graph.  For every connected set $S\subseteq V$, let
\begin{equation}
W(S)
:=
\sum_{xy\in E(G[S])}w_{xy}.
\label{eq:internal-weight}
\end{equation}
For disjoint sets $A,B\subseteq V$, write
\begin{equation}
w(A,B)
:=
\sum_{\substack{xy\in E\\x\in A,\ y\in B}}w_{xy}.
\label{eq:cut-weight}
\end{equation}

Fix $\alpha\in(0,1/2]$.  A bipartition
$S=A\mathbin{\dot\cup}B$ is called a connected
$\alpha$-balanced fragmentation of $S$ if $G[A]$ and $G[B]$ are
connected and
\begin{equation}
\max\{|A|-1,|B|-1\}
\le
(1-\alpha)(|S|-1).
\label{eq:shifted-balanced-fragmentation}
\end{equation}
The shifted cardinalities in \eqref{eq:shifted-balanced-fragmentation}
match the shifted Dasgupta objective.  For a connected set $S$ with
$|S|\ge2$, define
\begin{align}
\operatorname{frag}_{\alpha}(S)
&:=
\min_{\substack{S=A\mathbin{\dot\cup}B\\
G[A],G[B]\text{ connected}\\
\max\{|A|-1,|B|-1\}\le(1-\alpha)(|S|-1)}}
\frac{(|S|-1)w(A,B)}{W(S)}.
\label{eq:local-fragmentation-number}
\end{align}
If no admissible fragmentation exists, the minimum in
\eqref{eq:local-fragmentation-number} is defined to be $+\infty$.
The denominator is positive because $G[S]$ is connected and all support
weights are positive.

For a binary hierarchy $T$, let $\mathcal I(T)$ denote its family of
internal clusters.

Put
\begin{equation}
N:=n-1,
\qquad
q:=1-\alpha,
\qquad
J:=\left\lfloor\log_{1/q}N\right\rfloor.
\label{eq:fragmentation-scale-parameters}
\end{equation}
For $j\in\{0,\ldots,J\}$, define the hereditary fragmentation profile
\begin{align}
\kappa_j(G,w;\alpha)
&:=
\sup_{\substack{S\subseteq V,\ G[S]\text{ connected}\\
q^{j+1}N<|S|-1\le q^jN}}
\operatorname{frag}_{\alpha}(S).
\label{eq:hereditary-fragmentation-profile}
\end{align}
The supremum over an empty family is $0$.  Unlike the hierarchy-dependent
energy condition in Section~\ref{sec:energy}, this profile does not refer
to a hierarchy or to an optimal solution.

\begin{theorem}[Hereditary weighted fragmentation]
\label{thm:hereditary-fragmentation}
If $\kappa_j(G,w;\alpha)<\infty$ for every
$j\in\{0,\ldots,J\}$, then there exists a connected binary hierarchy
$T$ such that
\begin{equation}
\operatorname{cost}_{D}^{-}(T)
\le
W\sum_{j=0}^{J}\kappa_j(G,w;\alpha).
\label{eq:fragmentation-tree-cost}
\end{equation}
Consequently,
\begin{equation}
\frac{\OPTcard(G,w)}{\OPTfree(G,w)}
\le
\sum_{j=0}^{J}\kappa_j(G,w;\alpha).
\label{eq:fragmentation-gap-bound}
\end{equation}
\end{theorem}

\begin{IEEEproof}
Construct $T$ recursively.  At every current connected cluster $S$,
choose a connected $\alpha$-balanced fragmentation
$S=A_S\mathbin{\dot\cup}B_S$ attaining
$\operatorname{frag}_{\alpha}(S)$, and recurse on the two children.
The finiteness assumption ensures that the required fragmentation exists
at every nonsingleton cluster.  The construction therefore produces a
binary hierarchy whose clusters are connected.

For every $j\in\{0,\ldots,J\}$, let
\begin{equation}
\mathcal I_j(T)
:=
\left\{
S\in\mathcal I(T):
q^{j+1}N<|S|-1\le q^jN
\right\}.
\label{eq:fragmentation-scale-family}
\end{equation}
The members of $\mathcal I_j(T)$ form an antichain.  Indeed, if $R$ is
a proper descendant of $S$, then the first child of $S$ on the path to
$R$ and \eqref{eq:shifted-balanced-fragmentation} give
\begin{equation}
|R|-1
\le
q(|S|-1).
\label{eq:descendant-size-contraction}
\end{equation}
If $S\in\mathcal I_j(T)$, the right-hand side of
\eqref{eq:descendant-size-contraction} is at most $q^{j+1}N$, so $R$
cannot belong to the same scale family.  Since clusters in a hierarchy
are laminar, the sets in $\mathcal I_j(T)$ are therefore pairwise
disjoint.

At a cluster $S\in\mathcal I_j(T)$, the cost created by the selected
fragmentation satisfies
\begin{equation}
(|S|-1)w(A_S,B_S)
\le
\kappa_j(G,w;\alpha)W(S).
\label{eq:node-fragmentation-bound}
\end{equation}
Every support edge crosses the two children of exactly one internal
cluster.  Grouping edges by that cluster and then by scale gives
\begin{align}
\operatorname{cost}_{D}^{-}(T)
&=
\sum_{j=0}^{J}
\sum_{S\in\mathcal I_j(T)}
(|S|-1)w(A_S,B_S)
\nonumber\\
&\le
\sum_{j=0}^{J}\kappa_j(G,w;\alpha)
\sum_{S\in\mathcal I_j(T)}W(S)
\nonumber\\
&\le
W\sum_{j=0}^{J}\kappa_j(G,w;\alpha).
\label{eq:fragmentation-antichain-charge}
\end{align}
The last inequality holds because the clusters at a fixed scale are
pairwise disjoint, so their induced internal edge sets are disjoint.
This proves \eqref{eq:fragmentation-tree-cost}.  Finally,
\eqref{eq:free-lower-total-weight} gives $W\le\OPTfree(G,w)$, and
Corollary~\ref{cor:dasgupta-geometric-equivalence} identifies
$\OPTcard(G,w)$ with the minimum shifted cost over connected binary
hierarchies.  Minimizing the left-hand side of
\eqref{eq:fragmentation-tree-cost} therefore proves
\eqref{eq:fragmentation-gap-bound}.
\end{IEEEproof}

Theorem~\ref{thm:hereditary-fragmentation} converts a collection of
flat, locally chosen cuts into a compatible hierarchy.  Its scale sum
also distinguishes uniform regularity from stronger decay.

\begin{corollary}[Uniform and decaying fragmentation]
\label{cor:fragmentation-profiles}
Suppose that every connected $S\subseteq V$ with $|S|\ge2$ admits a
connected $\alpha$-balanced fragmentation.  If
\begin{equation}
\operatorname{frag}_{\alpha}(S)\le K,
\label{eq:uniform-fragmentation}
\end{equation}
then
\begin{equation}
\frac{\OPTcard(G,w)}{\OPTfree(G,w)}
\le
K\left(
1+\left\lfloor
\log_{1/(1-\alpha)}(n-1)
\right\rfloor
\right).
\label{eq:uniform-fragmentation-gap}
\end{equation}
More generally, if $\eta>0$ and
\begin{equation}
\operatorname{frag}_{\alpha}(S)
\le
\frac{K}{(|S|-1)^{\eta}},
\label{eq:decaying-fragmentation}
\end{equation}
then
\begin{equation}
\frac{\OPTcard(G,w)}{\OPTfree(G,w)}
\le
\frac{K(1-\alpha)^{-\eta}}
{1-(1-\alpha)^{\eta}}.
\label{eq:constant-fragmentation-gap}
\end{equation}
If instead $\theta>0$ and
\begin{equation}
\operatorname{frag}_{\alpha}(S)
\le
K(|S|-1)^{\theta},
\label{eq:growing-fragmentation}
\end{equation}
then
\begin{equation}
\frac{\OPTcard(G,w)}{\OPTfree(G,w)}
\le
\frac{K(n-1)^{\theta}}
{1-(1-\alpha)^{\theta}}.
\label{eq:polynomial-fragmentation-gap}
\end{equation}
\end{corollary}

\begin{IEEEproof}
Under \eqref{eq:uniform-fragmentation}, every term in
\eqref{eq:fragmentation-gap-bound} is at most $K$, which proves
\eqref{eq:uniform-fragmentation-gap}.  Put $q=1-\alpha$ and $N=n-1$.
For \eqref{eq:decaying-fragmentation}, every nonempty scale satisfies
\begin{equation}
\kappa_j(G,w;\alpha)
\le
K(q^{j+1}N)^{-\eta}.
\label{eq:decaying-scale-bound}
\end{equation}
Because $q^J N\ge1$ and $q^{J+1}N\ge q$, reversing the order of the
scales in \eqref{eq:decaying-scale-bound} bounds their sum by
$Kq^{-\eta}\sum_{r\ge0}q^{\eta r}$, which proves
\eqref{eq:constant-fragmentation-gap}.  Under
\eqref{eq:growing-fragmentation}, the $j$th scale contributes at most
$K(q^jN)^{\theta}$.  Summing over $j\ge0$ gives
\eqref{eq:polynomial-fragmentation-gap}.
\end{IEEEproof}

The next specialization gives a concrete local weight condition.  It
also shows that the logarithmic conclusion in
\eqref{eq:uniform-fragmentation-gap} is generally unavoidable.

\begin{lemma}[Balanced edge in a bounded-degree tree]
\label{lem:balanced-tree-edge}
Let $H$ be a tree on $m\ge2$ vertices with maximum degree at most
$\Delta\ge2$.  There is an edge whose removal produces components
$A$ and $B$ satisfying
\begin{equation}
\max\{|A|-1,|B|-1\}
\le
\left(1-\frac{1}{\Delta}\right)(m-1).
\label{eq:bounded-degree-tree-balance}
\end{equation}
\end{lemma}

\begin{IEEEproof}
Choose a centroid vertex $c$, so every component of $H-c$ has at most
$m/2$ vertices.  Let $A$ be a largest such component and let $B$ be
its complement.  Since $c$ has degree at most $\Delta$,
\begin{equation}
\frac{m-1}{\Delta}
\le
|A|
\le
\frac{m}{2}.
\label{eq:centroid-largest-branch}
\end{equation}
The edge from $c$ to $A$ separates $A$ from $B$.  The lower bound in
\eqref{eq:centroid-largest-branch} gives
\begin{equation}
|B|-1
=
m-|A|-1
\le
\left(1-\frac{1}{\Delta}\right)(m-1),
\end{equation}
and the upper bound gives the same inequality for $|A|-1$.
\end{IEEEproof}

\begin{corollary}[Locally regular weighted trees]
\label{cor:locally-regular-weighted-trees}
Let $G$ be a weighted tree of maximum degree at most $\Delta\ge2$.
Suppose that, for every connected $S\subseteq V$ with $|S|\ge2$, some
edge $e_S\in E(G[S])$ satisfying
\eqref{eq:bounded-degree-tree-balance} also satisfies
\begin{equation}
w_{e_S}
\le
\Lambda\frac{W(S)}{|S|-1}.
\label{eq:local-balanced-edge-regularity}
\end{equation}
Then
\begin{equation}
\frac{\OPTcard(G,w)}{\OPTfree(G,w)}
\le
\Lambda\left(
1+
\left\lfloor
\log_{\Delta/(\Delta-1)}(n-1)
\right\rfloor
\right).
\label{eq:weighted-tree-log-gap}
\end{equation}
In particular, \eqref{eq:weighted-tree-log-gap} holds if
\begin{equation}
\frac{\max_{e\in E}w_e}{\min_{e\in E}w_e}
\le
\Lambda.
\label{eq:tree-weight-aspect-ratio}
\end{equation}
\end{corollary}

\begin{IEEEproof}
Lemma~\ref{lem:balanced-tree-edge} permits $\alpha=1/\Delta$.  Since a
tree cut induced by one edge has weight $w_{e_S}$,
\eqref{eq:local-balanced-edge-regularity} implies
$\operatorname{frag}_{1/\Delta}(S)\le\Lambda$.  Apply
Corollary~\ref{cor:fragmentation-profiles}.  Under
\eqref{eq:tree-weight-aspect-ratio}, every connected $S$ has exactly
$|S|-1$ internal edges, and hence
\begin{equation}
\frac{W(S)}{|S|-1}
\ge
\min_{e\in E}w_e.
\end{equation}
Any balanced edge supplied by Lemma~\ref{lem:balanced-tree-edge} then
satisfies \eqref{eq:local-balanced-edge-regularity}.
\end{IEEEproof}

\begin{corollary}[Explicit construction on locally regular trees]
\label{cor:tree-algorithm}
Under the hypotheses of Corollary
\ref{cor:locally-regular-weighted-trees}, put
\begin{equation}
H_{n,\Delta}
:=
1+
\left\lfloor
\log_{\Delta/(\Delta-1)}(n-1)
\right\rfloor.
\label{eq:tree-recursion-depth}
\end{equation}
A hierarchy $T$ satisfying
\begin{equation}
\operatorname{cost}_{D}^{-}(T)
\le
\Lambda H_{n,\Delta} W
\le
\Lambda H_{n,\Delta}\OPTfree(G,w)
\le
\Lambda H_{n,\Delta}\OPTcard(G,w)
\label{eq:tree-constructive-bound}
\end{equation}
can be constructed in $O(nH_{n,\Delta})$ time.  In particular, the
running time is $O(n\log n)$ for fixed $\Delta$.
\end{corollary}

\begin{IEEEproof}
At a current connected subtree $S$, root $G[S]$ arbitrarily and compute
$W(S)$ and all rooted-subtree sizes in one traversal.  For every edge
$e\in E(G[S])$, these sizes give the cardinalities of the two components
of $G[S]-e$.  Among the edges satisfying
\eqref{eq:bounded-degree-tree-balance}, select one of minimum weight.
The hypothesis guarantees that this selected edge also satisfies
\eqref{eq:local-balanced-edge-regularity}.  Delete it and recurse on the
two components.

Every child $R$ of $S$ obeys
\begin{equation}
|R|-1
\le
\left(1-\frac{1}{\Delta}\right)(|S|-1),
\label{eq:tree-algorithm-contraction}
\end{equation}
so the recursion depth is at most $H_{n,\Delta}$.  At any fixed depth,
the current subtrees are vertex disjoint; hence all scans at that depth
take $O(n)$ time.  This proves the stated running time.  Charging the cut
edge chosen at $S$ by $(|S|-1)w_{e_S}\le\Lambda W(S)$ and applying the
same scale antichain argument as in Theorem
\ref{thm:hereditary-fragmentation} gives the first inequality in
\eqref{eq:tree-constructive-bound}.  The remaining inequalities follow
from $W\le\OPTfree(G,w)\le\OPTcard(G,w)$.
\end{IEEEproof}

\begin{theorem}[Tight logarithmic gap on bounded-degree trees]
\label{thm:binary-tree-log-gap}
There is a sequence of unweighted trees $B_h$ of maximum degree $3$,
with $n_h=2^{h+1}-1$ vertices, such that
\begin{equation}
\frac{\OPTcard(B_h)}{\OPTfree(B_h)}
=
\Theta(\log n_h).
\label{eq:binary-tree-theta-log}
\end{equation}
Every connected induced subgraph of $B_h$ satisfies
$\operatorname{frag}_{1/3}(S)\le1$.
\end{theorem}

\begin{IEEEproof}
Let $B_h$ be the complete binary tree of height $h$, with all support
edges assigned weight $1$.  Every connected induced subgraph of $B_h$
is a tree of maximum degree at most $3$.  Lemma
\ref{lem:balanced-tree-edge} supplies a $1/3$-balanced edge cut.  Since
$W(S)=|S|-1$ and the cut has weight $1$, its local fragmentation number
is at most $1$.  Corollary~\ref{cor:fragmentation-profiles} gives the
$O(\log n_h)$ upper bound in \eqref{eq:binary-tree-theta-log}.

For the lower bound, Corollary
\ref{cor:dasgupta-geometric-equivalence} allows an optimal cardinality
hierarchy to be chosen connected.  In a connected hierarchy over a
support tree, the two children of every internal cluster are separated
by exactly one support edge.  Grouping each support edge at the unique
internal cluster where its endpoints are first separated therefore gives
the following identity.  Here $\operatorname{depth}_{T}(v)$ denotes the
number of edges on the path from the hierarchy root to the leaf $v$, with
the root at depth $0$:
\begin{align}
\operatorname{cost}_{D}^{-}(T)
&=
\sum_{S\in\mathcal I(T)}(|S|-1)
\nonumber\\
&=
\sum_{v\in V(B_h)}\operatorname{depth}_{T}(v)
-(n_h-1).
\label{eq:tree-cost-external-depth}
\end{align}
Kraft's inequality for the $n_h$ leaves of the binary hierarchy gives
\begin{equation}
\sum_{v\in V(B_h)}\operatorname{depth}_{T}(v)
\ge
n_h\log_2 n_h.
\label{eq:external-depth-lower-bound}
\end{equation}
Consequently,
\begin{equation}
\OPTcard(B_h)
\ge
n_h\log_2 n_h-(n_h-1).
\label{eq:binary-tree-card-lower}
\end{equation}

It remains to upper-bound the free optimum.  Recursively build a
connected hierarchy on the two height-$(h-1)$ child subtrees, merge the
root vertex with the first child subtree, and then merge the resulting
cluster with the second child subtree.  The two top-level support edges
are each assigned at a cluster of diameter at most $2h$.  If $F_h$
denotes the diameter cost of this hierarchy, then
\begin{equation}
F_h
\le
2F_{h-1}+4h,
\qquad
F_0=0.
\label{eq:binary-tree-free-recurrence}
\end{equation}
Therefore
\begin{align}
\OPTfree(B_h)
&\le
F_h
\le
4\sum_{r=1}^{h}2^{h-r}r
\nonumber\\
&<
8\cdot2^h
=
4(n_h+1).
\label{eq:binary-tree-free-upper}
\end{align}
Combining \eqref{eq:binary-tree-card-lower} and
\eqref{eq:binary-tree-free-upper} proves the $\Omega(\log n_h)$ lower
bound and completes the proof.
\end{IEEEproof}

\begin{remark}[Oracle form beyond trees]
\label{rem:fragmentation-algorithm}
For a general support graph, Theorem
\ref{thm:hereditary-fragmentation} is constructive only relative to a
balanced-connected-cut oracle.  If the oracle returns a cut whose
normalized weight is within a factor $\beta$ of
\eqref{eq:local-fragmentation-number}, recursive calls produce the same
hierarchy with an additional factor $\beta$ in
\eqref{eq:fragmentation-gap-bound}.  We make no claim here that this
oracle problem is polynomial-time solvable on arbitrary graphs;
Corollary~\ref{cor:tree-algorithm} is the explicit algorithmic
specialization used in this paper.
\end{remark}

\section{Energy Decomposition}
\label{sec:energy}

The sharp theorem determines the worst-case gap.  We next identify the
quantity that controls smaller, instance-dependent gaps.

Let $T\in\mathcal T_{\mathrm{conn}}^{\mathrm{bin}}(G)$, and retain the
notation $\mathcal I(T)$ for its internal clusters.  For
$S\in\mathcal I(T)$,
write $S^0$ and $S^1$ for its two children and define
\begin{align}
b_T(S)
&:=
\sum_{\substack{xy\in E\\x\in S^0,\ y\in S^1}}w_{xy},
\label{eq:cut-mass}\\
\Delta_T(S)
&:=
\operatorname{diam}_{d_G}(S),
\label{eq:cluster-diameter}\\
a_T(S)
&:=
\Delta_T(S)b_T(S).
\label{eq:free-energy}
\end{align}
The diameter objective associated with $T$ is
\begin{equation}
F(T)
:=
\sum_{xy\in E}w_{xy}
\operatorname{diam}_{d_G}\bigl(C_T(x,y)\bigr).
\label{eq:tree-free-cost}
\end{equation}

\begin{lemma}[Cut-energy decomposition]
\label{lem:cut-energy-decomposition}
Every connected binary hierarchy $T$ satisfies
\begin{align}
F(T)
&=
\sum_{S\in\mathcal I(T)}\Delta_T(S)b_T(S)
=
\sum_{S\in\mathcal I(T)}a_T(S),
\label{eq:free-cut-decomposition}\\
\operatorname{cost}_{D}^{-}(T)
&=
\sum_{S\in\mathcal I(T)}(|S|-1)b_T(S).
\label{eq:dasgupta-cut-decomposition}
\end{align}
\end{lemma}

\begin{IEEEproof}
Every graph edge $xy$ crosses the two children of exactly one internal
cluster, namely $C_T(x,y)$.  Grouping edges by this unique cluster gives
both identities.
\end{IEEEproof}

For every internal cluster $S$, define its cardinality-to-diameter ratio
by
\begin{equation}
\rho_T(S)
:=
\frac{|S|-1}{\Delta_T(S)}.
\label{eq:density-ratio}
\end{equation}
Connectivity gives $\Delta_T(S)\le |S|-1$, while
$\Delta_T(S)\ge1$.  Therefore
\begin{equation}
1\le \rho_T(S)\le n-1.
\label{eq:density-range}
\end{equation}
Since $F(T)>0$, the free-energy distribution
\begin{equation}
\mu_T(S)
:=
\frac{a_T(S)}{F(T)}
\qquad
\bigl(S\in\mathcal I(T)\bigr)
\label{eq:energy-measure}
\end{equation}
is a probability measure on the internal clusters with positive energy.

\begin{lemma}[Gap identity on a fixed hierarchy]
\label{lem:gap-identity}
For every $T\in\mathcal T_{\mathrm{conn}}^{\mathrm{bin}}(G)$,
\begin{equation}
\frac{\operatorname{cost}_{D}^{-}(T)}{F(T)}
=
\sum_{S\in\mathcal I(T)}\mu_T(S)\rho_T(S)
=
\mathbb E_{\mu_T}[\rho_T].
\label{eq:gap-expectation}
\end{equation}
In particular, if $T_\star$ is free optimal in the normal form of
Corollary~\ref{cor:connected-diameter-representation}, then
\begin{equation}
\frac{\OPTcard(G,w)}{\OPTfree(G,w)}
\le
\mathbb E_{\mu_{T_\star}}[\rho_{T_\star}].
\label{eq:free-optimal-profile-bound}
\end{equation}
\end{lemma}

\begin{IEEEproof}
Substitute $|S|-1=\rho_T(S)\Delta_T(S)$ into
\eqref{eq:dasgupta-cut-decomposition} and divide by
\eqref{eq:free-cut-decomposition}.  For the final assertion, evaluate
the cardinality objective on the feasible hierarchy $T_\star$ and use
$F(T_\star)=\OPTfree(G,w)$.
\end{IEEEproof}

Combining Lemma~\ref{lem:gap-identity} with
Theorem~\ref{thm:sharp-gap} yields the useful two-scale estimate
\begin{equation}
\frac{\OPTcard(G,w)}{\OPTfree(G,w)}
\le
\min\left\{
\frac{2n-1}{3},
\mathbb E_{\mu_{T_\star}}[\rho_{T_\star}]
\right\}.
\label{eq:combined-gap-bound}
\end{equation}

The following elementary tail formulation records when this weighted
mean is bounded.  It uses the free-energy distribution appearing in the
objective, rather than a distribution obtained from uniformly sampled
pairs.

\begin{theorem}[Energy-tail transfer]
\label{thm:energy-profile-transfer}
Let $T\in\mathcal T_{\mathrm{conn}}^{\mathrm{bin}}(G)$ satisfy
\begin{equation}
F(T)\le \alpha\,\OPTfree(G,w)
\label{eq:approximate-free-tree}
\end{equation}
for some $\alpha\ge1$.  Suppose that a function
$\Phi:[1,n-1]\to\mathbb R_{\ge0}$ satisfies
\begin{equation}
\mu_T\bigl(\{S:\rho_T(S)>t\}\bigr)
\le
\Phi(t)
\qquad
\text{for every }t\in[1,n-1].
\label{eq:energy-tail-assumption}
\end{equation}
Then
\begin{equation}
\OPTcard(G,w)
\le
\alpha
\left(
1+\int_1^{n-1}\Phi(t)\,dt
\right)
\OPTfree(G,w).
\label{eq:energy-profile-conclusion}
\end{equation}
\end{theorem}

\begin{IEEEproof}
For the nonnegative random variable $\rho_T$, the layer-cake identity
and \eqref{eq:density-range} give
\begin{align}
\mathbb E_{\mu_T}[\rho_T]
&=
\int_0^\infty
\mu_T\bigl(\{S:\rho_T(S)>t\}\bigr)\,dt
\nonumber\\
&=
1+
\int_1^{n-1}
\mu_T\bigl(\{S:\rho_T(S)>t\}\bigr)\,dt
\nonumber\\
&\le
1+\int_1^{n-1}\Phi(t)\,dt.
\label{eq:layer-cake-bound}
\end{align}
Lemma~\ref{lem:gap-identity} gives
$\operatorname{cost}_{D}^{-}(T)=F(T)\mathbb E_{\mu_T}[\rho_T]$.
By Corollary~\ref{cor:dasgupta-geometric-equivalence},
$\OPTcard(G,w)$ is no larger than the shifted cardinality cost of $T$.
Thus \eqref{eq:approximate-free-tree} and
\eqref{eq:layer-cake-bound} prove the claim.
\end{IEEEproof}

\begin{corollary}[Polynomial tail]
\label{cor:polynomial-tail}
Under the hypotheses of Theorem~\ref{thm:energy-profile-transfer}, if
\begin{equation}
\Phi(t)\le A t^{-(1+\eta)}
\label{eq:polynomial-tail}
\end{equation}
for constants $A>0$ and $\eta>0$, then
\begin{equation}
\OPTcard(G,w)
\le
\alpha\left(1+\frac{A}{\eta}\right)
\OPTfree(G,w).
\label{eq:constant-profile-gap}
\end{equation}
If instead $\Phi(t)\le A/t$, the same argument gives
\begin{equation}
\OPTcard(G,w)
\le
\alpha\bigl(1+A\log(n-1)\bigr)
\OPTfree(G,w).
\label{eq:logarithmic-profile-gap}
\end{equation}
\end{corollary}

\begin{IEEEproof}
Under \eqref{eq:polynomial-tail},
$\int_1^{n-1}\Phi(t)\,dt\le A\int_1^\infty t^{-(1+\eta)}\,dt=A/\eta$.
If $\Phi(t)\le A/t$, then
$\int_1^{n-1}\Phi(t)\,dt\le A\log(n-1)$.
Both conclusions follow from
Theorem~\ref{thm:energy-profile-transfer}.
\end{IEEEproof}

\begin{remark}[Scope of the energy condition]
\label{rem:energy-scope}
Theorem~\ref{thm:energy-profile-transfer} is a first-moment consequence
of the exact identity \eqref{eq:gap-expectation}, not a new metric
embedding theorem.  Its tail premise is a hierarchy-dependent diagnostic
sufficient condition.  We do not claim that it is a canonical or weakest
graph regularity assumption, and we do not derive it here from a
stochastic model.
\end{remark}

\section{Geometric Consequences and Obstructions}
\label{sec:geometric-consequences}

The energy profile is weight-sensitive.  A simpler, purely geometric
bound follows by controlling the same ratio pointwise.

\begin{corollary}[Connected-set density bound]
\label{cor:connected-set-density}
Define
\begin{equation}
\Gamma(G)
:=
\max_{\substack{S\subseteq V,\ |S|\ge2\\G[S]\text{ connected}}}
\frac{|S|-1}{\operatorname{diam}_{d_G}(S)}.
\label{eq:geometric-density}
\end{equation}
Then
\begin{equation}
\OPTcard(G,w)
\le
\Gamma(G)\,\OPTfree(G,w).
\label{eq:geometric-density-gap}
\end{equation}
\end{corollary}

\begin{IEEEproof}
Every cluster of a free optimal connected hierarchy satisfies
$\rho_{T_\star}(S)\le\Gamma(G)$.  Apply
\eqref{eq:free-optimal-profile-bound}.
\end{IEEEproof}

The parameter $\Gamma(G)$ is equivalent, up to a factor of $2$, to a
linear ball-growth constant.  Define
\begin{equation}
\mathcal B(G)
:=
\max_{\substack{x\in V\\1\le r\le\operatorname{diam}_{d_G}(V)}}
\frac{|B_G(x,r)|-1}{r}.
\label{eq:ball-growth-constant}
\end{equation}
Every graph ball is connected.  A connected set $S$ of diameter $r$ is
contained in $B_G(x,r)$ for every $x\in S$, while
$\operatorname{diam}_{d_G}(B_G(x,r))\le2r$.  Consequently,
\begin{equation}
\Gamma(G)\le\mathcal B(G)\le2\Gamma(G).
\label{eq:ball-density-equivalence}
\end{equation}
In particular, the uniform growth bound $|B_G(x,r)|\le1+Lr$ implies
\begin{equation}
\OPTcard(G,w)\le L\,\OPTfree(G,w).
\label{eq:linear-growth-gap}
\end{equation}
For a path, $\Gamma(G)=1$, and the free and cardinality-constrained
optima coincide.

Table~\ref{tab:representative-instances} compares the two structural
parameters on representative unit-weight support graphs.  Statements in
the fragmentation column hold for every connected induced subgraph $S$,
except that the displayed clique value is stated for the root set
$S=V$.

\begin{table*}[t]
\caption{Representative unit-weight instances.  Here
$n_h=2^{h+1}-1$ for the complete binary tree $B_h$, with $h\ge1$.}
\label{tab:representative-instances}
\centering
\footnotesize
\begin{tabular}{|p{0.14\textwidth}|p{0.21\textwidth}|p{0.29\textwidth}|p{0.24\textwidth}|}
\hline
Support graph & Connected-set density & Fragmentation behavior & Cardinality/free ratio \\
\hline
Path $P_n$
& $\Gamma(P_n)=1$
& $\operatorname{frag}_{1/2}(S)=1$
& $\OPTcard/\OPTfree=1$ \\
\hline
Cycle $C_n$
& $\Gamma(C_n)\le2$
& $\operatorname{frag}_{1/2}(S)\le2$
& $\OPTcard/\OPTfree\le2$ \\
\hline
Complete binary tree $B_h$
& $\Gamma(B_h)\ge(n_h-1)/(2h)$
& $\operatorname{frag}_{1/3}(S)\le1$
& $\OPTcard/\OPTfree=\Theta(\log n_h)$ \\
\hline
Clique $K_n$
& $\Gamma(K_n)=n-1$
& $\operatorname{frag}_{1/2}(V)=2\lfloor n^2/4\rfloor/n=\Theta(n)$
& $\OPTcard/\OPTfree=(2n-1)/3$ \\
\hline
\end{tabular}
\end{table*}

For $C_n$, every proper connected induced subgraph is a path, while the
whole cycle has a balanced two-edge cut; this proves the fragmentation
entry, and linear ball growth proves the density entry.  Thus the cycle
is a non-tree sparse family with a constant gap.  In $B_h$, the root set
has diameter $2h$, so its density is at least $(n_h-1)/(2h)$, whereas
Theorem~\ref{thm:binary-tree-log-gap} gives a unit fragmentation profile
and a logarithmic gap.  This shows that the pointwise density bound can be
far weaker than recursive fragmentation.  Finally, a balanced cut of
$K_n$ has $\lfloor n^2/4\rfloor$ crossing edges, which gives the displayed
root fragmentation value; the same dense family realizes the sharp
linear gap in Theorem~\ref{thm:sharp-gap}.

Conditions involving only ties or the numerical aspect ratio of the
similarity weights cannot yield a constant gap.

\begin{proposition}[Distinct weights do not prevent a linear gap]
\label{prop:distinct-weight-complete}
For every $n\ge2$ and every $\delta>0$, the complete graph $K_n$
admits pairwise distinct edge weights in $[1,1+\delta]$ for which
\begin{equation}
\frac{\OPTcard(G,w)}{\OPTfree(G,w)}
\ge
\frac{2n-1}{3(1+\delta)}.
\label{eq:distinct-weight-complete-gap}
\end{equation}
\end{proposition}

\begin{IEEEproof}
Choose distinct weights in $[1,1+\delta]$.  Since the support metric of
$K_n$ is the discrete metric, the constant ultrametric with value $1$
on distinct pairs is feasible for the free problem.  Together with
\eqref{eq:free-lower-total-weight}, this gives
\begin{equation}
\OPTfree(K_n,w)=W.
\label{eq:weighted-complete-free}
\end{equation}
For every binary hierarchy $T$, all edge weights are at least $1$, so
Lemma~\ref{lem:pair-cardinality-identity} implies
\begin{equation}
\operatorname{cost}_{D}^{-}(T)
\ge
\frac{n(n-1)(2n-1)}{6}.
\label{eq:weighted-complete-card-lower}
\end{equation}
On the other hand,
$W\le(1+\delta)\binom{n}{2}$.  Combining this inequality with
\eqref{eq:weighted-complete-free} and
\eqref{eq:weighted-complete-card-lower} proves
\eqref{eq:distinct-weight-complete-gap}.
\end{IEEEproof}

\section{Discussion and Limitations}
\label{sec:discussion}

Theorem~\ref{thm:sharp-gap} completely determines the worst-case cost of
cardinality realizability.  It also shows why a universal constant
comparison is impossible.  Theorem~\ref{thm:hereditary-fragmentation}
complements this negative result with a tree-independent sufficient
condition: light connected balanced cuts can be assembled recursively,
and their scale-wise normalized weights control the resulting hierarchy.
A uniform hereditary bound yields an $O(\log n)$ comparison, and
Theorem~\ref{thm:binary-tree-log-gap} shows that this conclusion is tight
even for unweighted support trees of maximum degree $3$.  Stronger decay
of the fragmentation profile is sufficient for a constant comparison.

The regularity hypothesis is local and hereditary rather than merely a
condition on the multiset of weights.  Proposition
\ref{prop:distinct-weight-complete} shows that eliminating ties and
making the global weight aspect ratio arbitrarily close to $1$ still
permits a linear gap.  By contrast,
\eqref{eq:local-balanced-edge-regularity} succeeds on bounded-degree
trees because it couples weight mass to a balanced separator inside
every connected induced subtree.  This distinction is important when
formulating probabilistic or rank-based variants of the condition.

The energy identity gives a different, hierarchy-sensitive description:
a large gap occurs when a significant fraction of the free objective is
paid at connected clusters whose cardinality is large relative to their
graph diameter.  The tail corollary is an elementary first-moment
reformulation of this distributional control.  It is included as a
diagnostic for a supplied hierarchy and is not an additional
graph-structural theorem or a claimed weakest regularity principle.

The comparison in Lemma~\ref{lem:connectedification} with the original
free ultrametric is deliberately edgewise.  For a nonedge
$xy\notin E$, connected refinement may delay the merger of $x$ and
$y$, so $u_T^{\mathrm{diam}}(x,y)\le u(x,y)$ need not hold.  This
causes no loss because the objective is supported on $E$.
Independently, $u_T^{\mathrm{diam}}\ge d_G$ holds on every pair in
$V\times V$.

Finally, throughout the paper $d_G$ is the unit-length metric of the
support graph, while $w$ supplies similarity mass only to the objective.
Consequently, changing a positive weight does not change $d_G$, whereas
removing that edge can change the metric discontinuously.  The geometric
interpretation should therefore be understood with the support graph and
its similarity weights specified as separate parts of the input.
When every pair has positive similarity, the support graph is complete
and $d_G$ is the discrete metric; in that dense regime the free benchmark
reduces to $W$ and the geometric structure is correspondingly limited.
The interpretation is most informative for sparse or deliberately
thresholded similarity graphs, including neighborhood graphs.  We do not
claim that the support metric is a canonical conversion of similarity
magnitudes into distances.

\section{Conclusion}
\label{sec:conclusion}

Dasgupta's shifted objective is the minimum weighted average stretch
into cardinality-realizable ultrametrics that dominate the support-graph
metric.  Connectedification puts this constrained problem and the free
dominating-ultrametric problem on the same hierarchy class.  Their
worst-case ratio is exactly $(2n-1)/3$, while the standard unshifted
objective has sharp factor $2(n+1)/3$.  A hereditary weighted
fragmentation profile supplies a constructive, tree-independent gap
bound: uniform local control gives $O(\log n)$, this order is tight on
bounded-degree unweighted trees, and polynomial scale decay gives a
constant.  On individual hierarchies, the gap is also exactly the
free-energy-weighted mean of
$(|S|-1)/\operatorname{diam}_{d_G}(S)$.  This identity supplies an
auxiliary diagnostic for instance-dependent bounds; the paper's
structural guarantee is the hereditary fragmentation theorem.  On
locally regular bounded-degree trees, its recursive construction runs in
$O(n\log n)$ time.
These statements compare geometric benchmarks and do not imply a new
approximation guarantee for a linkage heuristic.


\begin{thebibliography}{16}

\bibitem{Dasgupta2016}
S.~Dasgupta, ``A cost function for similarity-based hierarchical
clustering,'' in \emph{Proc. 48th Annu. ACM Symp. Theory of Computing
(STOC)}, 2016, pp.~118--127, doi: 10.1145/2897518.2897527.

\bibitem{CarlssonMemoli2010}
G.~Carlsson and F.~M\'emoli, ``Characterization, stability and
convergence of hierarchical clustering methods,'' \emph{J. Mach. Learn.
Res.}, vol.~11, no.~47, pp.~1425--1470, 2010. [Online]. Available:
\url{https://www.jmlr.org/papers/v11/carlsson10a.html}

\bibitem{RoyPokutta2017}
A.~Roy and S.~Pokutta, ``Hierarchical clustering via spreading metrics,''
\emph{J. Mach. Learn. Res.}, vol.~18, no.~88, pp.~1--35, 2017.

\bibitem{MoseleyWang2023}
B.~Moseley and J.~R. Wang, ``Approximation bounds for hierarchical
clustering: Average linkage, bisecting $k$-means, and local search,''
\emph{J. Mach. Learn. Res.}, vol.~24, pp.~1--36, 2023.  A preliminary
version appeared in \emph{Advances in Neural Information Processing
Systems 30}, 2017.

\bibitem{ABN2011}
I.~Abraham, Y.~Bartal, and O.~Neiman, ``Advances in metric embedding
theory,'' \emph{Adv. Math.}, vol.~228, no.~6, pp.~3026--3126, 2011,
doi: 10.1016/j.aim.2011.08.003.

\bibitem{ABN2015}
I.~Abraham, Y.~Bartal, and O.~Neiman, ``Embedding metrics into
ultrametrics and graphs into spanning trees with constant average
distortion,'' \emph{SIAM J. Comput.}, vol.~44, no.~1, pp.~160--192,
2015, doi: 10.1137/120884390.

\bibitem{CharikarChatziafratis2017}
M.~Charikar and V.~Chatziafratis, ``Approximate hierarchical clustering
via sparsest cut and spreading metrics,'' in \emph{Proc. 28th Annu.
ACM--SIAM Symp. Discrete Algorithms (SODA)}, 2017, pp.~841--854,
doi: 10.1137/1.9781611974782.53.

\bibitem{CohenAddadEtAl2019}
V.~Cohen-Addad, V.~Kanade, F.~Mallmann-Trenn, and C.~Mathieu,
``Hierarchical clustering: Objective functions and algorithms,''
\emph{J. ACM}, vol.~66, no.~4, Art.~26, 2019,
doi: 10.1145/3321386.

\bibitem{CharikarChatziafratisNiazadeh2019}
M.~Charikar, V.~Chatziafratis, and R.~Niazadeh, ``Hierarchical
clustering better than average-linkage,'' in \emph{Proc. 30th Annu.
ACM--SIAM Symp. Discrete Algorithms (SODA)}, 2019, pp.~2291--2304,
doi: 10.1137/1.9781611975482.139.

\bibitem{HogemoEtAl2020}
S.~H\o gemo, C.~Paul, and J.~A. Telle, ``Hierarchical clusterings of
unweighted graphs,'' in \emph{Proc. 45th Int. Symp. Mathematical
Foundations of Computer Science (MFCS)}, ser. LIPIcs, vol.~170, 2020,
pp.~47:1--47:13, doi: 10.4230/LIPIcs.MFCS.2020.47.

\bibitem{HogemoEtAl2021}
S.~H\o gemo, B.~Bergougnoux, U.~Brandes, C.~Paul, and J.~A. Telle,
``On Dasgupta's hierarchical clustering objective and its relation to
other graph parameters,'' in \emph{Proc. 23rd Int. Symp. Fundamentals
of Computation Theory (FCT)}, ser. Lecture Notes in Computer Science,
vol.~12867, 2021, pp.~287--300,
doi: 10.1007/978-3-030-86593-1\_20.

\bibitem{ArutyunovaRoglin2025}
A.~Arutyunova and H.~R\"oglin, ``The price of hierarchical
clustering,'' \emph{Algorithmica}, vol.~87, pp.~1420--1452, 2025,
doi: 10.1007/s00453-025-01327-7.

\bibitem{Bartal1996}
Y.~Bartal, ``Probabilistic approximation of metric spaces and its
algorithmic applications,'' in \emph{Proc. 37th Annu. IEEE Symp.
Foundations of Computer Science (FOCS)}, 1996, pp.~184--193,
doi: 10.1109/SFCS.1996.548477.

\bibitem{FRT2004}
J.~Fakcharoenphol, S.~Rao, and K.~Talwar, ``A tight bound on
approximating arbitrary metrics by tree metrics,'' \emph{J. Comput.
Syst. Sci.}, vol.~69, no.~3, pp.~485--497, 2004,
doi: 10.1016/j.jcss.2004.04.011.

\bibitem{CohenAddadKarthikLagarde2020}
V.~Cohen-Addad, K.~C. S., and G.~Lagarde, ``On efficient low distortion
ultrametric embedding,'' in \emph{Proc. 37th Int. Conf. Machine
Learning (ICML)}, ser. Proc. Mach. Learn. Res., vol.~119, 2020,
pp.~2078--2088. [Online]. Available:
\url{https://proceedings.mlr.press/v119/cohen-addad20a.html}

\bibitem{AilonCharikar2011}
N.~Ailon and M.~Charikar, ``Fitting tree metrics: Hierarchical
clustering and phylogeny,'' \emph{SIAM J. Comput.}, vol.~40, no.~5,
pp.~1275--1291, 2011, doi: 10.1137/100806886.

\end{thebibliography}
\end{document}